\documentclass{article}

\usepackage{amssymb,amsfonts,amsthm}
\usepackage{amsmath}
\usepackage{bm}
\usepackage[mathscr]{eucal}
\usepackage{wasysym}

\def\pmb#1{\setbox0=\hbox{$#1$}%
  \kern-.025em\copy0\kern-\wd0
  \kern.05em\copy0\kern-\wd0
  \kern-.025em\raise.0433em\box0}
\def\pmbs#1{\setbox0=\hbox{$\scriptstyle #1$}%
  \kern-.0175em\copy0\kern-\wd0
  \kern.035em\copy0\kern-\wd0
  \kern-.0175em\raise.0303em\box0}
\def\be{\begin{equation}}
\def\ee{\end{equation}}
\def\bea{\begin{eqnarray}}
\def\eea{\end{eqnarray}}

\def\vec#1{\mbox{\boldmath$#1$}}

\def\Om{\Omega}

\def\bom{\mbox{\boldmath $\omega$}}

\def\bna{\mbox{\boldmath $\nabla$}}

\def\vece{\vec{e}}

\def\la{\langle}
\def\ra{\rangle}
\def\hsp5{\hspace{5mm}}
\def\case#1/#2{\textstyle\frac{#1}{#2}}

\renewcommand{\vector}[1]{\bm{#1}}
\newcommand{\textfrac}[2]{{\textstyle \frac{#1}{#2}}}

\newcommand{\nt}{\tilde{N}}
\newcommand{\Htilde}{\tilde{\mathscr{H}}}
\def\cg{{\cal G}}

\theoremstyle{plain}
\newtheorem{theorem}{Theorem}[section]
\newtheorem{corollary}[theorem]{Corollary}

\theoremstyle{remark}
\newtheorem*{remark}{Remark}

\def\Sp{\Sigma_+}
\def\Sb{\bar{\Sigma}}
\def\St{\tilde{\Sigma}}
\def\Sc{\check{\Sigma}}
\def\vsq{{v^2}}
\def\Omk{\Omega_\mathrm{k}}

\title{\sc Future asymptotics of tilted Bianchi type II cosmologies}
\author{
\sc Sigbj{\o}rn Hervik $^{1}$\thanks{Electronic address:
{\tt sigbjorn.hervik@uis.no}}\
,\ Woei Chet Lim $^{2}$\thanks{Electronic address:
{\tt wclim@aei.mpg.de}}\
,\ Patrik Sandin $^{3}$\thanks{Electronic address:
{\tt patrik.sandin@kau.se}}\ , \ and Claes Uggla$^{3}$\thanks{Electronic address:
{\tt claes.uggla@kau.se}}\\
$^{1}${\small\em Faculty of Science and Technology, University of Stavanger,}\\
{\small\em N-4036 Stavanger, Norway}\\
$^{2}${\small\em Albert-Einstein-Institut, Am M\"uhlenberg 1,}\\
{\small\em D-14476 Potsdam, Germany}\\
$^{3}${\small\em Department of Physics, University of Karlstad,}\\
{\small\em S-65188 Karlstad, Sweden}
}
\begin{document}
\maketitle
\begin{abstract}
In this paper we study the future asymptotics of spatially
homogeneous Bianchi type II cosmologies with a tilted perfect
fluid with a linear equation of state. By means of Hamiltonian
methods we first find a monotone function for a special tilted
case, which subsequently allows us to construct a new set of
monotone functions for the general tilted type II cosmologies.
In the context of a new partially gauge invariant dynamical
system, this then leads to a proof for a theorem that for the
first time gives a complete description of the future
asymptotic states of the general tilted Bianchi type II models.
The generality of our arguments suggests how one can produce
monotone functions that are useful for determining the
asymptotics of other tilted perfect fluid cosmologies, as well
as for other sources.
\end{abstract}

\section{Introduction}

Spatially homogeneous anisotropic perfect fluid models have
during the last decades been successfully studied using a
dynamical systems approach. The book~\cite{waiell97} summarizes
most of the presently known results about the so-called
non-tilted perfect fluid cosmologies, while the more general
`tilted' perfect fluid models have been primarily investigated
more recently~\cite{hew91}---\cite{colher08}.

In all of the papers investigating tilted models, the analysis
has rested on standard techniques from dynamical systems
theory. Most of the results concern the identification of fixed
points and a subsequent linear stability analysis of these
points. In order to get a grip on the global aspects of the
solutions, an effective tool is the use of monotone functions.
However, such monotone functions are hard to find, and in most
of the previous works on tilted models the monotone functions
were obtained by brute force, trial and error, and luck. It
would therefore be desirable to have a more systematic method
to seek and find such monotone functions.

For non-tilted spatially homogeneous perfect fluid models,
virtually all known results crucially rely on the existence
of conserved quantities and monotone functions. These have turned
out to be connected to the existence of certain symmetries,
intimately associated with conservation laws such
as the preservation of the number of particles in a fluid element,
and the so-called scale-automorphism group~\cite{heiugg10}.
Although not necessary, the symmetries and associated structures
were, to a large extent, found by means of Hamiltonian techniques,
see ch. 10 in~\cite{waiell97} and~\cite{heiugg10}. One aim
of this paper is to illustrate that one can fruitfully use similar
methods that previously have been applied to non-tilted models to
tilted ones. To do so, we will consider an example---the tilted
Bianchi type II models.

The tilted perfect fluid Bianchi type II models have been
analyzed before as a dynamical system in~\cite{hewetal01}. In
that paper, as well as in the present one, the perfect fluid
was assumed to obey a linear equation of state characterized by
$\tilde{p}/\tilde{\rho}=w=\textrm{const}$, where $\tilde{p}$
and $\tilde{\rho}$ are the pressure and energy density with
respect to the rest frame of the fluid, respectively; special
cases of interest are dust, $w=0$, radiation, $w=\frac{1}{3}$,
a stiff fluid, $w=1$, while a cosmological constant $\Lambda$
can formally be regarded as a perfect fluid with $w=-1$,
however, in this paper we will consider the range $-1<w<1$.

In~\cite{hewetal01} the fixed points of the system were found
and their linear stability properties were studied. It was
observed that the future stability of the fixed points depended
on the equation of state parameter $w$, but that a future
stable fixed point existed for all $w$ in the range $-1<w<1$
under consideration. On the basis of the linear analysis and
numerical computations it was conjectured that the future
linearly stable fixed points were global attractors in the full
state space. This conjecture was corroborated by means of
monotone functions for some ranges of $w$. However, the
complete picture was not fully drawn, mainly because of the
lack of a set of sufficiently restrictive monotone functions.
In this paper we will use methods that previously have been
applied to the non-tilted models and find a set of new monotone
functions that are sufficiently restrictive to determine the
future asymptotics of all tilted solutions, thus filling in the
missing gaps in the conclusions drawn in~\cite{hewetal01}.

The outline of the paper is as follows. In
Section~\ref{sec:BII} we give a new partially gauge invariant
dynamical system for the general tilted Bianchi type II models,
derived in Appendix~\ref{derdyn}. Then, based on the structure
of monotone functions obtained for more special models by means
of Hamiltonian methods, presented in Appendix~\ref{ham}, we
obtain the new monotone functions in Section~\ref{asymp}, which
leads to a new theorem that describes the future asymptotic
states of all tilted Bianchi type II cosmologies. Finally, in
Section~\ref{disc} we conclude with a discussion about the
structure of the monotone functions, why they exist, and why
one can hope to expect similar structures in other models.

\section{Dynamical systems description of tilted Bianchi type II cosmologies}
\label{sec:BII}

In~\cite{sanugg10} the so-called conformally Hubble-normalized
orthonormal frame equations are given in full generality. These
are in turn specialized to the spatially homogeneous Bianchi
case in Appendix~\ref{derdyn} and then to the presently studied
Bianchi type II models with a general tilted perfect fluid with
a linear equation of state. Then a subsequent set of new
variables, invariant under frame rotations in the 23-plane (see
Appendix~\ref{derdyn}), yield the following state vector and
dynamical system.

\vspace*{4mm}
\noindent {\em State vector}:
\be
{\bf S} = (\Sigma_+,\bar{\Sigma}, \tilde{\Sigma}^2,\check{\Sigma}^2,
\Omega_\mathrm{k},v^2)\, ,
\ee
where we treat $\tilde{\Sigma}^2,\,\Sc^2$, and $v^2$ as variables, but where we have
refrained from giving them new names.

\vspace*{4mm} \noindent {\em Evolution equations}:
\begin{subequations}\label{dynsys}
\begin{align}
\Sigma_+^\prime &= -(2-q)\Sigma_+ - 3\Sc^2\ + 4\Omk + \tfrac12(1+w)G_+^{-1}\vsq\Omega\, ,\\
\Sb^\prime &= -(2-q)\Sb -2\sqrt3\St^2 +\sqrt3\Sc^2  +\tfrac{\sqrt3}{2}(1+w)G_+^{-1}\vsq\Omega\, ,\\
(\St^2)^\prime &= -2(2-q-2\sqrt3\Sb)\St^2\, ,\\
(\Sc^2)^\prime &= -2[2-q - 3\Sigma_+ + \sqrt3\Sb]\Sc^2\, ,\\
\Omk^\prime &= 2(q - 4\Sigma_+)\Omk\, ,\\
(v^2)^\prime &= 2G_-^{-1}\,(1- \vsq)[3w-1 - \Sigma_+ - \sqrt3\Sb]\vsq\, .
\end{align}

\vspace*{4mm} \noindent {\em Constraint equation}:
\be\label{codazzi} f({\bf S}) = 4\Sc^2\Omk - (1+w)^2G_+^{-2}\vsq\Omega^2 = 0\, .
\ee
\end{subequations}

In the above equations $\Omega$ is given by the Gauss constraint
\be\label{gauss} \Omega = 1 - \Sigma^2 - \Omk\, , \ee
where
\be \Sigma^2 = \Sigma_+^2 + \Sb^2 + \St^2 + \Sc^2\, . \ee
The condition $\Omega\geq 0$ in combination with~\eqref{gauss}
yields $0\leq \Sigma_+^2 + \Sb^2 + \St^2 + \Sc^2 + \Omk\leq 1$.
The deceleration parameter $q$ is given by
\be q = 2\Sigma^{2} + \textfrac{1}{2}G_+^{-1}[1+3w +
(1-w)v^2]\Omega\, , \ee
while
\be
G_\pm = 1 \pm wv^2\, ;
\ee
finally, ${}^\prime$ denotes differentiation with respect to a
dimensionless time parameter $\tau$, determined by $d\tau = H
dt$, where $t$ is the clock time along the congruence normal to
the spatially homogeneous hypersurfaces.

We now give a brief description of the invariant subsets and
fixed points of the system~\eqref{dynsys}, which is analogous
to the analysis given by Hewitt et al.~\cite{hewetal01}. Note that
although our system is invariant under frame rotations in the 23-plane
it is not invariant under all rotations. Hence there exist multiple
representations of solutions, for further comments on this,
see~\cite{hewetal01}.

\begin{table}[h!]
\begin{tabular}{lllcl}
\hline
& Name & Restrictions & Dimension & Interior/Boundary\\
\hline \vspace{-3.5mm}\\
i) & Non-tilted non-vacuum Bianchi type II & $v^2 = \check{\Sigma}^2 = 0$ & 4 &
Boundary\\
ii) & Non-tilted non-vacuum Bianchi type I & $v^2 = \Omk = 0$ & 4 &
Boundary\\
iii) & Vacuum Bianchi type II & $\check{\Sigma}^2 = \Omega = 0$ & 4 &
Boundary\\
iv) & Vacuum Bianchi type I (Kasner) & $\Omk = \Omega = 0$ & 4 &
Boundary\\
v) & Extreme tilt & $v^2 = 1$ & 4 &
Boundary\\
vi) & Orthogonally transitive Bianchi type II &
$\tilde{\Sigma}^2 = 0$ & 4 & Interior\\
\hline
\end{tabular}
\caption{Invariant sets of the state space. The last column
indicates if a subset is part of the boundary of the
general tilted Bianchi type II models or if it is an interior subset.}
\label{invariant sets}
\end{table}

\vspace*{4mm} \noindent {\em Fixed points}:
\begin{itemize}
\item[i)] The flat Friedmann solution, $\mathrm{F}$: $\quad -1 < w <
    1$,
\begin{displaymath}
\Sigma_+=\bar{\Sigma}=\tilde{\Sigma}^2 = \check{\Sigma}^2 = \Omk = v^2 = 0\, .
\end{displaymath}
\item[ii)] The Collins-Stewart solution, $\mathrm{CS}$: $\quad -
    \frac{1}{3} < w < 1$,
\begin{displaymath}
\Sigma_+=\frac{1}{8}(3w+1)\,,\quad
\bar{\Sigma}=\tilde{\Sigma}^2 = \check{\Sigma}^2 = v^2 = 0\, ,\quad \Omk =
\frac{3}{64}(3w + 1)(1 - w)\, .
\end{displaymath}
\item[iii)] Hewitt's solution, $\mathrm{H}$: $\quad  \textfrac{3}{7} <
    w < 1$,
\begin{displaymath}
\Sigma_+=\frac{1}{8}(3w+1)\, ,\quad \bar{\Sigma}=
\frac{\sqrt{3}}{8}(7w - 3)\, ,\quad \tilde{\Sigma}^2 = 0\, ,\quad
\check{\Sigma}^2 = \frac{3(1-w)(11w+1)(7w-3)}{16(17w-1)}\, ,
\end{displaymath}
\begin{displaymath}
\Omk= \frac{3(1-w)(5w+1)(3w-1)}{4(17w-1)}\, , \quad v^2 =
\frac{(3w-1)(7w-3)}{(11w+1)(5w+1)}\, .
\end{displaymath}
\item[iv)] Hewitt et al.'s 1-parameter set of
    solutions,\footnote{There is a misprint
    in~\cite{hewetal01} where the square root on the $b$ in
    the $\Sigma_1$ expression of the line of fixed points,
    HL, has disappeared.} $\mathrm{HL}$: $\quad
    w=\frac{5}{9},\quad 0 < b=\mathrm{const} < 1$,
\begin{displaymath}
\Sigma_+=\frac{1}{3}\, ,\quad \bar{\Sigma}=
\frac{1}{3\sqrt{3}}\, ,\quad \tilde{\Sigma}^2 =
\frac{4}{27}\,b\, ,\quad \check{\Sigma}^2 =
\frac{4(4b+1)(8-3b)}{513}\, ,
\end{displaymath}
\begin{displaymath}
\Omk = \frac{(2b+1)(17-8b)}{171}, \quad v^2
= \frac{3(4b+1)(2b+1)}{(17-8b)(8-3b)}\, .
\end{displaymath}
\item[v)] Hewitt et al.'s extreme tilted point, $\mathrm{Het}$:
    $\quad  -1 < w < 1$,
\begin{displaymath}
\Sigma_+=\frac{1}{3}\, ,\quad \bar{\Sigma}=
\frac{1}{3\sqrt{3}}\, ,\quad \tilde{\Sigma}^2 =
\frac{4}{27}\, ,\quad \check{\Sigma}^2 =
\frac{100}{513}\, , \quad \Omk = \frac{3}{19},\quad v^2 = 1\, .
\end{displaymath}
\end{itemize}
The system also admits the following fixed point sets:
\emph{the Kasner circle} $\mathrm{K}^\ocircle$, for which
$v^2=0$, \emph{the Kasner lines} $\mathrm{KL}_\pm$, for which
$v^2=\mathrm{const}$, and \emph{the extremely tilted Kasner
circle} $\mathrm{Ket}^\ocircle$, for which $v^2=1$. These
subsets reside on the Bianchi type I vacuum boundary, i.e.,
$\Sigma^2=1$, with $\tilde{\Sigma}=\check{\Sigma}=0$,
see~\cite{hewetal01}, however, since these fixed points do not
play a prominent role in this paper we refrain from giving them
explicitly.

\begin{remark}
Below we will refer to the relevant fixed point values for $\Sp$
and $\Sb$ by $\Sigma_{+0}$ and ${\Sb}_0$, respectively.
\end{remark}
%

\section{Future asymptotes in tilted Bianchi type II cosmology}\label{asymp}

In what follows certain monotone functions will play a crucial role.
Based on our results in Appendix~\ref{derdyn} and~\ref{ham}, we
hence begin by deriving them.

\subsection{Monotone functions}\label{mon}

There are several auxiliary equations that are useful in the context
of monotonic functions, see Appendix~\ref{derdyn}:

\vspace*{4mm} \noindent {\em Auxiliary equations}:
\begin{subequations}\label{aux}
\begin{align}
\Omega^\prime &= [2q - (1 + 3w) + (1+w)(3w-1 - \Sigma_+ - \sqrt{3}\bar{\Sigma})G_+^{-1} v^2]\Omega\, ,\\
Q^\prime &= -[2(1-q) + \Sigma_+ + \sqrt{3}\bar{\Sigma}]Q\, ,\\
\Psi^\prime &= [2q - (1+3w)]\Psi\, , \label{aux3}
\end{align}
\end{subequations}
where
\be\label{Psidef} Q = (1+w)G_+{-1}\,v\,\Omega\, ,\qquad \Psi =
\Gamma^{-(1-w)}\,G_+^{-1}\,\Omega\, .
\ee

Since
\be 2q - (1+3w) = 4\Sigma^2 - (1 + 3w)(1-\Omega) +
(1-3w)(1+w)G_+^{-1}v^2\Omega \geq 0\, ,\quad \text{if}\,
-1<w\leq -1/3\, ,
\ee
as follows from~\eqref{decpar}, $\Psi$ is a monotonically increasing function
when $-1<w\leq -1/3$; henceforth we denote $\Psi$ in this interval of the
equation of state by $M_\mathrm{F}$.

Before we continue, let us introduce some notation:
\begin{subequations}
\begin{align}
\phi_* &= 1-\Sigma_{+0}\Sp-\Sb_0\Sb\, ,\\
\varphi_* &= [\Sigma_{+0}(\Sp-\Sigma_{+0})+\Sb_0(\Sb-\Sb_0)]^2\, ,\\
\bar{\varphi}_* &= [\Sb_0(\Sp-\Sigma_{+0})-\Sigma_{+0}(\Sb-\Sb_0)]^2\, ,
\end{align}
\end{subequations}
where the subscript $_*$ henceforth denotes a specific fixed
point, while $\Sigma_{+0}$ and $\Sb_0$ are the associated fixed
point values for $\Sigma_+$ and $\Sb$, respectively. In the
following it is important that $\phi_*>0$, which can be seen as
follows:
\be
\phi_* = \tfrac12\left[
1 - \Sigma_{+0}^2 - \Sb_{0}^2 + 1 - \Sigma_+^2 -  \Sb^2 +
(\Sigma_+ - \Sigma_{+0})^2 +  (\Sb -\Sb_{0})^2 \right]
> \tfrac12\left[1 - \Sigma_{+0}^2 - \Sb_{0}^2\right] >0\, ,
\ee
where we have used the Gauss constraint~\eqref{gauss} and $\Omega>0,\,\Omk >0$.

For non-tilted perfect fluid models
\be
M_\mathrm{CS:v^2=0} = \phi_\mathrm{CS}^{-2}\,\Omega_k^m\,\Omega^{1-m}
= \frac{\Omega_k^m\,\Omega^{1-m}}{(1-\Sigma_{+0}\Sigma_+)^2}\, ;\quad
m = \frac{3(1-w)\Sigma_{+0}}{8(1-\Sigma_{+0}^2)}\, ,\,\, \Sigma_{+0} = \frac{1}{8}(1+3w)\, ,
\quad \Sb_0=0\, ,
\ee
is a monotonically increasing function. However, it is
of interest to generalize this function by replacing $\Omega$ with $\Psi$, which
is equal to $\Omega$ in the non-tilted case, i.e.,
\be
M_\mathrm{CS} = \phi_\mathrm{CS}^{-2}\,\Omega_k^m\,\Psi^{1-m}\, ,
\ee
which leads to the following time derivative in the present fully tilted state space:
\be (\ln M_\mathrm{CS})'= 3\phi_\mathrm{CS}^{-1}\left[(1-w)
\left(\frac{\varphi_\mathrm{CS}}{(1-\Sigma_{+0}^2)\Sigma_{+0}^2}
+ \frac{\bar{\varphi}_\mathrm{CS}}{\Sigma_{+0}^2} +
\St^2\right) + \frac18 (3-7w)(2\check{\Sigma}^2 + (1+w)G_+
^{-1}\,v^2\Omega)\right]\, , \ee
and hence $M_\mathrm{CS}$ is monotonically increasing when
$-1/3 < w \leq 3/7$.

In Appendix~\ref{ham} we derive a monotonic function for the tilted
orthogonally transitive case $\tilde{\Sigma}^2=0$ which can be written as
\be
M_\mathrm{H} =
\phi_\mathrm{H}^{-(3+13w)}(\check{\Sigma}^2)^{\frac{7w-3}{2}}\Omk^{3w-1}\Psi^4\, .
\ee
It sometimes turns out to be the case that a monotone function
for a given state space is also monotone in a more general
state space in which the original is embedded in, at least for
a limited range of the equation of state parameter, see
e.g.~\cite{heretal06,heiugg10}. We hence compute the time
derivative for $M_\mathrm{H}$ in the full tilted case; this
gives us

\be
(\ln M_\mathrm{H})' =\phi_\mathrm{H}^{-1}\left[
\frac{49(16\varphi_\mathrm{H} + 3(1-w)(3+13w)\bar{\varphi}_\mathrm{H})}{8+\frac{3}{8}(91w-31)(7w-3)}
+ \frac{3}{4}(5-9w)(13w+3)\St^2\right]\, ,
\ee
and hence $M_\mathrm{H}$ is monotonically increasing when
$3/7<w\leq 5/9$.

The above monotonic functions all have the form
\be\label{monform}
M_{*}=\phi_*^{-\beta}\,(\St^2)^{\alpha_1}(\check{\Sigma}^2)^{\alpha_2}
\Omk^{\alpha_3}(\Psi)^{\alpha_4}\, ,
\ee
where $\beta=2(\alpha_1+\alpha_2+\alpha_3+\alpha_4)$, and
$\alpha_1$, $\alpha_2$, $\alpha_3$, $\alpha_4$, $\beta\geq0$.
For the individual cases we have
\begin{align}
M_\mathrm{F}:& \quad \alpha_1=\alpha_2=\alpha_3=0\, ,\ \alpha_4=1\, ,\ \beta=2\, ,\\
M_\mathrm{CS}:& \quad \alpha_1=\alpha_2=0\, ,\ \alpha_3=m\, ,\ \alpha_4=1-m,\ \beta=2\, ,\\
M_\mathrm{H}:& \quad \alpha_1=0\, ,\ \alpha_2=\tfrac12(7w-3)\,,\ \alpha_3=3w-1\,,\ \alpha_4=4\,,\
\beta=3+13w\, .
\end{align}

Let us assume the form~\eqref{monform} in order to find a monotonic function for the range
$5/9 <w <1$. We obtain
\be
M_\mathrm{Het} = \phi_\mathrm{Het}^{-46}(\St^2)^4(\check{\Sigma}^2)^{10}\Omk^9\, ,
\ee
and hence
$\alpha_1=4\,,\, \alpha_2=10\,,\, \alpha_3=9\,,\, \alpha_4=0\,,\, \beta=46$,
which leads to
\be
(\ln M_\mathrm{Het})' = \phi_\mathrm{Het}^{-1}
\left[243\varphi_\mathrm{Het} + 207\bar{\varphi}_\mathrm{Het} + \frac{23(9w-5)}{3}G_+^{-1}(1-v^2)\Omega\right]\, ,
\ee
and thus $M_\mathrm{Het}$ is monotonically increasing when $5/9 <w <1$.

\subsection{Future asymptotic limits}\label{future}

Let us denote the invariant set for the general tilted Bianchi type case, for
which $(1-v^2)\,v^2\,\tilde{\Sigma}^2\,\check{\Sigma}^2\,\Omega\,\Omk\neq 0$, by ${\cal S}_\mathrm{Gen}$,
while we denote the orthogonally transitive case, for which
$(1-v^2)\,v^2\,\check{\Sigma}^2\,\Omega\,\Omk\neq 0$, by ${\cal S}_\mathrm{OT}$.

A local analysis, as in~\cite{hewetal01}, reveals that the
state space ${\bf S}$ have fixed points as \emph{local sinks}
according to Table~\ref{sinks}, which leads to the following
bifurcation diagram:
\be
\mathrm{F} \xrightarrow{\,\, w=-1/3\,\,} \mathrm{CS} \xrightarrow{\,\, w=3/7\quad}
\mathrm{H} \xrightarrow{w=5/9,\, b=0} \mathrm{HL} \xrightarrow{w=5/9,\, b=1}
\mathrm{Het}\, .
\ee
\begin{table}[h!]
\begin{center}
\begin{tabular}{lc}
{\bf Range of} $w$ & {\bf Sink} \\
\hline\vspace{-3.5mm}\\
$-1 < w \leq -1/3$ & $\mathrm{F}$ \\
$-1/3 < w \leq 3/7$ & $\mathrm{CS}$ \\
$3/7 < w < 5/9$ & $\mathrm{H}$ \\
$w = 5/9$ & $\mathrm{HL}$ \\
$5/9 < w <1$ & $\mathrm{Het}$ \\
\hline
\end{tabular}
\caption{Sinks for ${\cal S}_\mathrm{Gen}$.} \label{sinks}
\end{center}
\end{table}

It was conjectured in~\cite{hewetal01} that the above local
results hold globally modulo a set of solutions of measure
zero. Using the above monotone functions one can show that this
is indeed the case, and also identify all the exceptional
solutions. In the following theorem we show that \emph{all}
solutions that belong to ${\cal S}_\mathrm{Gen}$ end up
asymptotically at the above sinks:
\begin{theorem}\label{thm}
For all ${\bf x}\in {\cal S}_\mathrm{Gen}$
\[
\omega({\bf x}) = \begin{cases}\mathrm{F} &  -1 < w \leq -1/3 \\
\mathrm{CS} & -1/3 < w \leq 3/7 \\
\mathrm{H} & 3/7 < w < 5/9 \\
\mathrm{HL} & w = 5/9  \\
\mathrm{Het} & 5/9 < w <1
\end{cases}\, .
\]
\end{theorem}

The proof makes use of the \emph{Monotonicity
Principle}~\cite{waiell97,lebetal95}, which gives information
about the global asymptotic behavior of solutions of a
dynamical system. It is stated as follows: Let $\phi_\tau$ be a
flow on $\mathbb{R}^n$ with ${\bf X}$ being an invariant set.
Furthermore, let $M$ be a ${\mathcal C}^1$ function $M:{\bf X}
\rightarrow \mathbb{R}$.
Then if $M$ is increasing on orbits, then for all $x\in X$
\begin{equation}
\omega(x) \subseteq
\{{\bf s} \in \overline{X}\backslash X\:|\: \lim\limits_{{\bf y}\rightarrow {\bf s}} M({\bf y}) \neq
\inf\limits_{X} M \}\, .
\end{equation}
\begin{proof}
We make use that ${\cal S}_\mathrm{Gen}$ is a relatively
compact set and hence that every orbit in ${\cal
S}_\mathrm{Gen}$ has an $\omega$-limit point in $\overline{\cal
S}_\mathrm{Gen}$, moreover, the $\omega$-limit set of every
orbit in ${\cal S}_\mathrm{Gen}$ must be an invariant set. In
every case $\inf\limits_{{\cal S}_\mathrm{Gen}} M_*=0$, and
hence we only have to investigate the set where $(\ln
M_*)^\prime =0$. It turns out that in all cases the invariant
set associated with $(\ln M_*)^\prime =0$ is precisely the
pertinent fixed point(s), which thus is the $\omega$-limit of
every orbit in ${\cal S}_\mathrm{Gen}$. Hence the proof, and
the situations, is virtually identical to that of the
non-tilted Bianchi type II perfect fluid case given
in~\cite{hewetal01} on p. 151. For the cases $-1 < w \leq
-1/3$, $-1/3 < w \leq 3/7$, $3/7 < w < 5/9$ and $w = 5/9$, $5/9
< w <1$, one uses $M_\mathrm{F}$, $M_\mathrm{CS}$,
$M_\mathrm{H}$, and $M_\mathrm{Het}$, respectively.
\end{proof}
\begin{remark}
It follows that $\mathrm{F}, \mathrm{CS}, \mathrm{H}, \mathrm{Het}$
attract a 4-parameter set of solutions, as does the line $\mathrm{HL}$,
however, in this case it follows from the reduction theorem, see e.g.~\cite{cra91},
that each point on the line attracts a 3-parameter set since the line is
transversally hyperbolic.
\end{remark}
\begin{corollary}
For all ${\bf x}\in {\cal S}_\mathrm{OT}$
\[
\omega({\bf x}) = \begin{cases}\mathrm{F} &  -1 < w \leq -1/3 \\
\mathrm{CS} & -1/3 < w \leq 3/7 \\
\mathrm{H} & 3/7 < w < 1
\end{cases}\, .
\]
\end{corollary}
\begin{proof}
This follows immediately from the previous proof, in combination with noticing
the form for $(\ln M_\mathrm{H})^\prime$ when $\tilde{\Sigma}^2=0$.
\end{proof}
We hence have established that the local bifurcation diagram
\be
\mathrm{F} \xrightarrow{w=-1/3} \mathrm{CS} \xrightarrow{\,\ w=3/7\,\,}
\mathrm{H}
\ee
from~\cite{hewetal01} reflects the global features of the solution space of
${\cal S}_\mathrm{OT}$.

\section{Discussion}\label{disc}

In this paper we have found that there exists a collection of
monotonically increasing functions that completely determine
and describe the future asymptotics of tilted Bianchi type II
models. Furthermore, they all take the form
\be\label{monform2}
M_{*}=\phi_*^{-\beta}\,(\St^2)^{\alpha_1}(\check{\Sigma}^2)^{\alpha_2}
\Omk^{\alpha_3}(\Psi)^{\alpha_4}\, ,
\ee
where $\beta=2(\alpha_1+\alpha_2+\alpha_3+\alpha_4)$, and
$\alpha_1$, $\alpha_2$, $\alpha_3$, $\alpha_4$, $\beta\geq0$.
Moreover, the time derivatives of the monotone functions
$M_\mathrm{CS},\,M_\mathrm{H},\,M_\mathrm{Het}$, with nonzero
$\Sigma_{+0}^2 + \Sb_0^2$ (and hence non-zero $\varphi_* +
\bar{\varphi}_*$) all take the form
\be
(\ln M_*)' = \phi_*^{-1}(A_*\varphi_* + B_*\bar{\varphi}_* + \mathrm{Inv}_*)\, ,
\ee
where $A_*$ and $B_*$ are constants, and $\mathrm{Inv}_*$ a
function that vanishes on one of the invariant sets in Table
\ref{invariant sets}. Outside these invariant sets the function
$\mathrm{Inv}_*$ is positive only in a limited range of values
for $w$. In the $M_\mathrm{CS}$ case $\mathrm{Inv}_\mathrm{CS}$
is zero for the non-tilted subset for which $v^2=\Sc^2=0$; in
the $M_\mathrm{H}$ case $\mathrm{Inv}_\mathrm{H}$ is zero for
the orthogonally transitive subset for which $\St^2=0$;
finally, in the $M_\mathrm{Het}$ case
$\mathrm{Inv}_\mathrm{Het}$ is zero for the extreme tilt subset
for which $v^2=1$. Clearly it is easier to find these monotone
functions on these subsets first and then extend them to the
larger state space, as we did in order to find the key monotone
function $M_\mathrm{H}$ in Appendix~\ref{ham}, moreover, we
also first computed $M_\mathrm{Het}$ for the extreme tilt
subset $v^2=1$, although we did not use Hamiltonian methods in
this case.\footnote{If one wants to use Hamiltonian methods to
deal with extreme tilt, one first has to observe that these
models have the same equations as those associated with a
source that takes the form of a perfect fluid with a radiation
equation of state $w=1/3$, but with a null vector field
replacing the timelike 4-velocity.}

The importance of the hierarchical structure of Bianchi cosmology,
where we have systems with boundaries on boundaries, have been emphasized
before, see e.g.~\cite{heietal09,heiugg09a,heiugg09b} and references therein.
Here we see yet another context for this observation, which suggests
that one should first try to find monotone functions for subsets and
then attack the case one is really interested in. Hence one should first
identify subsets for a given state space and write
them on the form $Z_A=0$ and then, if there exists a locally future stable
fixed point on a given subset that admits subsets $Z_a=0$, attempt to find
monotone functions of the form
\be\label{monform3}
M_{*}=\phi_*^{-\beta}\prod_a\, Z_a^{\alpha_a}\, .
\ee

Is there a deeper reason for why monotone functions like this should exist?
The analysis of the non-tilted case in~\cite{heiugg10} suggests that the existence
of these monotone functions are related to the scale-automorphism group.
In the tilted case this group can be viewed as consisting of an off-diagonal
special automorphism group and a diagonal scale-automorphism group.
The off-diagonal special automorphism give rise to conserved momenta,
if the underlying symmetry is not broken by source terms, and
hence also to monotone functions, in a similar way as for the non-tilted
models, see~\cite{heiugg10}. We have not discussed such monotone functions
here since we did not need them for the future asymptotics, however,
they could be of help for the much more difficult past asymptotic behavior.
However, the off-diagonal automorphisms also have other dynamical consequences.
It is because of the off-diagonal automorphisms we only had diagonal
shear degrees of freedom in $\phi_*$ in the present case, and it is
because of this one would expect the present analysis to also be of relevance
for other tilted models, and for other sources. Moreover, in the
present case we have encountered a hierarchy of source subsets that
is typical. It is the increasing complexity of the source that breaks
the vacuum symmetry group associated with the scale-automorphism group,
creating a hierarchy of monotone functions associated with different
source subsets. This is also what one encounters in the non-tilted
case~\cite{heiugg10}, however, with an increasingly complex source
this phenomenon seems to be even more pronounced! Hence a systematic
attempt on the tilted models, or other sources, strongly suggests
a deeper investigation into the dynamical consequences for the
scale-automorphism group for the various relevant subsets. This is
quite ambitious task and we have therefore refrained from doing it
here; instead we have made use of the structures one can expect to arise
from such an analysis without deriving all the details that completely
determines all the monotone functions from the scale-automorphism
group (for a hint of how this can be done, see the complete analysis
of the class A non-tilted models from a Hamiltonian perspective
in~\cite{heiugg10}) in order to see if this is likely to be a fruitful
project. The answer seems to be yes.

\subsection*{Acknowledgments}
CU is supported by the Swedish Research Council.

\begin{appendix}

\section{Derivation of the dynamical system}\label{derdyn}

In spatially homogeneous cosmology the space-time is foliated by
a geodesically parallel family of spatially homogeneous slices
with a timelike unit normal vector $n^a$, see,
e.g.,~\cite{waiell97,ellmac69,jan01} and references therein.

\subsection{Perfect fluids}\label{perf}

Splitting the total stress-energy tensor $T_{ab}$ with respect
to $n^a$ yields:
\begin{subequations}\label{Tirreduc}
\begin{align}
T_{ab} &= \rho\,n_a\,n_b + 2q_{(a}\,n_{b)} + p\,h_{ab} + \pi_{ab}\:,\\
\rho &= n^a\,n^b\,T_{ab}\:,\qquad q_a = -h_a{}^b\,n^c\,T_{bc}\:,\qquad
p = \textfrac{1}{3}\,h^{ab}\,T_{ab}\:,\qquad \pi_{ab} = h_{\la
a}{}^c\,h_{b\ra}{}^d\,T_{cd}\, ,
\end{align}
\end{subequations}
where $h_{ab} = n_a n_b + g_{ab}$ and
$A_{\la ab \ra} = h_a{}^c h_b{}^d A_{cd} - \frac{1}{3}h_{ab}h^{cd}A_{cd}$;
$\rho, p$ is the total
energy density and total effective pressure, respectively,
measured in the rest space of $n^a$.
In this paper we consider a perfect fluid, which yields the
stress-energy tensor:
\be T^{ab} = (\tilde{\rho} + \tilde{p}) \tilde{u}^a\tilde{u}^b
+ \tilde{p} g^{ab}\, , \ee
where $\tilde{\rho}$ and $\tilde{p}$ are the energy density and
pressure, respectively, in the rest frame of the fluid, while
$\tilde{u}^a$ is its 4-velocity; throughout we assume that
$\tilde{\rho}\geq 0$. Making a 3+1 split with respect to $n^a$,
leads to
\bea \tilde{u}^a = \Gamma(n^a + v^a)\,; \qquad n_a v^a=0\, ,\qquad
\Gamma = (1-v^2)^{-1/2}\, , \eea
where $v^a$ is the three-velocity of the fluid, also known as the
tilt vector; this gives
\be\label{nonormpfrel} \tilde{\rho} = \Gamma^{-2}\,G_+^{-1}\,\rho\,
,\quad q^a = (1 + w) G_+^{-1}\, \rho\, v^a,\quad\,\,
p = w\rho + \textfrac{1}{3} (1 - 3w)q_a v^a\, ,\quad\,\,
\pi_{ab} = q_{\la a}v_{b\ra}\, , \ee
where $G_\pm  =  1 \pm w\, v^2$, $w =\tilde{p}/\tilde{\rho}$.

\subsection{Orthonormal frame equations}\label{onframe}

In Bianchi cosmology the metric can be written as
\be\label{ds2} ^4{\bf g} = -N^2(x^0) \,dx^0 \otimes dx^0 +
g_{\hat{i}\hat{j}}(x^0)\:(\bom^{\hat{i}} +
N^{\hat{i}}\,dx^0)\otimes (\bom^{\hat{j}} + N^{\hat{j}}\,dx^0)
\qquad (\hat{i},\hat{j}=1,2,3)\, , \ee
where $\{\bom^{\hat{i}}\}$ is a left-invariant\footnote{See
\cite{waiell97} p. 38 on the meaning of group invariant frames
and their relation to orthonormal frames.} co-frame on $G$ dual
to a left-invariant spatial frame $\{\vector{e}_{\hat{i}}\}$.
This frame is a basis of the Lie algebra with structure
constants $C^{\hat{i}}{}_{\hat{j}\hat{k}}$, i.e.,
\be [\vector{e}_{\hat{i}},\vector{e}_{\hat{j}}] =
C^{\hat{k}}{}_{\hat{i}\hat{j}}\,\vector{e}_{\hat{k}} =
(\epsilon_{\hat{i}\hat{j}\hat{m}}\,n^{\hat{m}\hat{k}} +
2a_{[\hat{i}}\,\delta_{\hat{j}]}{}^{\hat{k}})\,\vector{e}_{\hat{k}}
\qquad \text{or, equivalently,} \qquad d\bom^{\hat{i}} =
-\textfrac{1}{2}C^{\hat{i}}{}_{\hat{j}\hat{k}}\,\bom^{\hat{j}}\wedge
\bom^{\hat{k}}\, , \ee
where we have decomposed the structure constants
$C^{\hat{i}}{}_{\hat{j}\hat{k}}$ as~\cite{estetal68}:
\be\label{Cdecomp} C^{\hat{i}}{}_{\hat{j}\hat{k}} =
\epsilon_{\hat{j}\hat{k}\hat{m}}\,n^{\hat{m}\hat{i}} +
a_{\hat{n}}\,\delta^{\hat{n}\hat{i}}_{\hat{j}\hat{k}}\, ,\qquad
a_{\hat{i}} = \textfrac{1}{2} C^{\hat{j}}{}_{\hat{i}\hat{j}}\, ,
\ee
where $a_{\hat{i}}\,a_{\hat{j}} =
\textfrac{1}{2}h\epsilon_{\hat{i}\hat{k}\hat{m}}\epsilon_{\hat{j}\hat{p}\hat{n}}\,
n^{\hat{k}\hat{p}}n^{\hat{m}\hat{n}}$, $a^2 =
a_{\hat{i}}\,a^{\hat{i}} =
\textfrac{1}{2}h[(n^{\hat{i}}{}_{\hat{i}})^2 -
n^{\hat{i}}{}_{\hat{j}}\,n^{\hat{j}}{}_{\hat{i}}]$, where $h$
is a constant group invariant that is unaffected by frame
choices. The Bianchi models are divided into two main classes:
The class A models for which $a_{\hat{i}}=0$, and the class B
models for which $a_{\hat{i}}\neq 0$.

For simplicity we will set the shift vector to zero, i.e.,
$N^{\hat{i}}=0$. This leads to the orthonormal frame $^4{\bf g}
= -(\bom^0)^2 + (\bom^1)^2 + (\bom^2)^2 + (\bom^3)^2$, and a
frame $\{\vector{e}_\alpha\}$ that is dual to
$\{\bom^\alpha\}$, $\alpha=1,2,3$,
\be\label{framerel} \vector{e}_0 =
N^{-1}\frac{\partial}{\partial x^0}\,, \qquad \vector{e}_\alpha =
e_\alpha{}^{\hat{i}}(x^0)\,\hat{\vector{e}}_{\hat{i}} =
e_\alpha{}^{\hat{i}}(x^0)\,e_{\hat{i}}{}^i\,\frac{\partial}{\partial
x^i} \quad\text{where} \quad \delta^{\alpha\beta}
e_\alpha{}^{\hat{i}}(x^0)\,e_\beta{}^{\hat{j}}(x^0) =
g^{\hat{i}\hat{j}}(x^0)\, ,\ee
where $g^{\hat{i}\hat{j}}(x^0)$ is the left-invariant
contravariant spatial metric associated with
$g_{\hat{i}\hat{j}}(x^0)$ ($e_{\hat{i}}{}^i =
e_{\hat{i}}{}^i(x^1,x^2,x^3)$).

Since the unit normal to the spatial symmetry surfaces ${\bf n}
= \vector{e}_0$ by definition is hypersurface forming, and
since it is the tangent to a geodesic congruence due to
spatial homogeneity, we obtain
\begin{subequations}\label{comm}
\begin{align} \label{dcomts0} [\,\vece_{0}, \vece_{\alpha}\,] &=
C^\beta{}_{0\alpha}\,\vece_{\beta} = f_{\alpha}{}^{\beta}\,\vece_{\beta} = -
[\,H\,\delta_{\alpha}{}^{\beta} + \sigma_{\alpha}{}^{\beta} +
\epsilon_{\alpha}{}^{\beta}{}_{\gamma}\,\Omega^{\gamma}\,]
\vece_{\beta}\, , \\
\label{dcomtsa} [\,\vece_{\alpha}, \vece_{\beta}\,] &=
C^\gamma{}_{\alpha\beta}\,\vece_{\gamma} =
\left[2a_{[\alpha}\,\delta_{\beta]}{}^{\gamma} +
\epsilon_{\alpha\beta\delta}\,n^{\delta\gamma}\right]\vece_{\gamma}\, ,
\end{align}
\end{subequations}
where $H$ is the Hubble variable; $\sigma_{\alpha\beta}$ is the
shear associated with ${\bf n}$; $\Omega^\alpha$ is the Fermi
rotation which describes how the spatial triad rotates with
respect to a gyroscopically fixed so-called Fermi
frame.\footnote{The sign in the definition of $\Omega^\alpha$
is the same as in~\cite{ellels99,sanugg10}, but opposite of
that in~\cite{waiell97}.} The relations~\eqref{framerel}
and~\eqref{comm} yield
\begin{subequations}\label{onquant}
\begin{alignat}{3}
\label{H} H &=
\textfrac{1}{3}\left(g^{-\frac{1}{2}}\right)\vece_{0}\left(g^{\frac{1}{2}}\right) =
-\textfrac{1}{3}e^\alpha{}_{\hat{i}}\,\vece_0\left(e_\alpha{}^{\hat{i}}\right)\, ,
& \qquad
\sigma_{\alpha\beta} &= - e^\gamma{}_{\hat{i}}\,\delta_{\gamma\la\alpha}\,
\vece_0\left(e_{\beta\ra}{}^{\hat{i}}\right)\, ,&\qquad \\
\Omega^\alpha &=
\textfrac{1}{2}\epsilon^\alpha{}_\beta{}^\gamma\,e^\beta{}_{\hat{i}}\,
\vece_0\left(e_\gamma{}^{\hat{i}}\right)\, ,
&\qquad
n^{\alpha\beta} &=
g^{-\frac{1}{2}}e^\alpha{}_{\hat{i}}\,e^\beta{}_{\hat{j}}\,n^{\hat{i}\hat{j}}\, ,
&\qquad
a_\alpha &= e_\alpha{}^{\hat{i}}\,a_{\hat{i}}\, ,\label{n}
\end{alignat}
\end{subequations}
where 
$g$ is the determinant of the spatial metric $g_{\hat{i}\hat{j}}$, and
$e^\alpha{}_{\hat{i}}$ is the inverse of
$e_\alpha{}^{\hat{i}}$, i.e.,
\be g=\det(g_{\hat{i}\hat{j}}) = (\det(e^\alpha{}_{\hat{i}}))^2
= (\det(e_\alpha{}^{\hat{i}}))^{-2} ,\qquad
e^\alpha{}_{\hat{i}}\,e_\beta{}^{\hat{i}}
=\delta^\alpha{}_\beta\, . \ee

The matter conservation equation $\bna_a T^{ab} = 0$ for a perfect fluid
with a linear equation of state yield
\begin{subequations}
\begin{align}
({\rm ln}\,\rho)\,\dot{} &= (1+w) G_+^{-1}[-3H +
f_{\alpha\beta}\, v^\alpha v^\beta + 2a_{\alpha}v^{\alpha}],\label{rhoperf}\\
\dot{v} &= G_-^{-1}\,(1-v^2)\,
\left[3w H + f_{\alpha\beta}\,c^\alpha c^\beta - 2w\,a_\beta\,c^\beta\,v \right] v ,\label{pec}\\
\dot{c}_{\alpha} &= [\delta_\alpha{}^\beta - c_\alpha c^\beta]
[f_\beta{}^\gamma\,c_{\gamma} - v (a_\beta +
\epsilon_{\beta\gamma\delta}\,
n^{\delta\zeta}\,c_{\zeta}\,c^{\gamma})]\, , \label{cpec}
\end{align}
\end{subequations}
where $a_\beta + \epsilon_{\beta\gamma\delta}\,
n^{\delta\zeta}\,c_{\zeta}\,c^{\gamma} =
C^\zeta{}_{\beta\gamma}\,c_{\zeta}\,c^{\gamma}$, and where
instead of the 3-velocity $v_\alpha$ we have found it
convenient to introduce $v = \sqrt{v_\alpha v^\alpha}\geq 0$
and the unit vector $c_\alpha=v_\alpha/v$ as variables. To
obtain a more compact notation, we also introduced $\dot{f} =
(f)\,\dot{} = \vece_0 f = N^{-1}d f/d x^0 = df/d t$, where $t$
is the clock time associated with the normal congruence of the
spatial symmetry surfaces (i.e. $N=1$ for this
parameterization).

It is of interest to consider the particle number density
$\tilde{n}$ and the chemical potential $\tilde{\mu}$, which,
for a linear equation of state, can be defined as (see~\cite{jan01,jan83}
and references therein)
\be\label{nmudef} \tilde{\rho} = \tilde{n}^{1+w}, \qquad
\tilde{\mu} = (1+w)\,\tilde{n}^w\, .
\ee
Defining
\be\label{ldef} l = \tilde{n}\,g^{\frac{1}{2}}\,\Gamma\, , \ee
yields the evolution equation $({\rm ln}\, l)\,\dot{} =
2a_{\alpha}v^{\alpha}  = 2(a_{\alpha}c^{\alpha})\,v$, and hence
$l$ is a constant of the motion whenever
$a_{\alpha}v^{\alpha}=0$, e.g. for the class A perfect fluid
models. Another quantity of interest is Taub's spatial
circulation 1-form~\cite{taub69,elsugg97} $t_a = \tilde{\mu}\,
\tilde{u}_a$, whose spatial components can be written as
$t_\alpha = \tilde{\mu}\,\Gamma\, v_\alpha =
\tilde{\mu}\,\Gamma\, v\,c_\alpha$, with the norm
$\tilde{\mu}\,\Gamma\, v$, which satisfies $(\ln\,
\tilde{\mu}\,\Gamma\, v)\,\dot{} =
f_{\alpha\beta}\,c^\alpha\,c^\beta$, which, together
with~\eqref{cpec} yields $\dot{t}_\alpha =
\tilde{\mu}\,\Gamma\, v\,(f_{\alpha\beta} -
v\,C^\gamma{}_{\alpha\beta}\,c_\gamma)\,c^\beta$. This in turn
leads to that $t_{\hat{i}} = e^\alpha{}_{\hat{i}}\,t_\alpha$,
where
$\vece_0\left(e^\alpha{}_{\hat{i}}\right) = -
f_\beta{}^\alpha\,e^\beta{}_{\hat{i}}$, obeys the equation
\be\label{tcons} \dot{t}_{\hat{i}} =-(\tilde{\mu}\,\Gamma)^{-1}\,
C^{\hat{k}}{}_{\hat{i}\hat{j}}\,t_{\hat{k}}\,t^{\hat{j}}\, . \ee
%

\subsection{The Hubble-normalized dynamical systems approach}\label{hubdynsys}

In the conformal Hubble normalized approach one factors out the Hubble
variable $H$ by means of a conformal transformation which
yields dimensionless quantities~\cite{rohugg05,sanugg10}. In the
spatially homogeneous case this amounts to the following:
\be (\Sigma_{\alpha\beta}, R^\alpha, N^{\alpha\beta},A_\alpha)
= \frac{1}{H}(\sigma_{\alpha\beta}, \Omega^\alpha,
n^{\alpha\beta},a_\alpha)\, ,\qquad (\Omega, P, Q_\alpha,
\Pi_{\alpha\beta}) =
\frac{1}{3H^2}(\rho,p,q_\alpha,\pi_{\alpha\beta})\, ,\ee
where we have chosen to normalize the stress-energy quantities
with $3H^2$ rather than $H^2$ in order to conform with the
usual definition of $\Omega$; in the perfect fluid case this
leads to that
\be\label{pfrel} Q^\alpha = (1 + w)G_+^{-1}\, v\,\Omega\,
c^\alpha\, ,\quad Q = (1 + w)G_+^{-1}\,v\,\Omega\, ,\quad P =
w\Omega + \textfrac{1}{3}(1 - 3w)v\,Q\, ,\quad\,\,
\Pi_{\alpha\beta} = v\,Q\, c_{\la\alpha}c_{\beta\ra}\, . \ee

In addition to this we choose a new dimensionless time variable
$\tau$ by means of the lapse choice $N=H^{-1}$.
Since $H$ is the only variable with dimension, its evolution
equation decouples from the remaining equations for dimensional
reasons:
\be\label{Heq} H^\prime = -(1+q)H\, ;\qquad\qquad  q =
2\Sigma^{2} + \textfrac{1}{2}(\Om+3P)\, ,\qquad
\Sigma^2=\textfrac{1}{6}\Sigma_{\alpha\beta}\Sigma^{\alpha\beta}\,,
\ee
where a prime denotes $d/d\tau$ and where $q$ is the
deceleration parameter, obtained by means of one of Einstein's
equations $G_{ab}=T_{ab}$ ---the Raychaudhuri equation (we use
units $c=1$ and $8\pi G=1$, where $c$ is the speed of light and
$G$ is Newton's gravitational constant); in the perfect fluid
case we obtain
\be\label{decpar} q = 2\Sigma^2 + \textfrac{1}{2}(\Om+3P) = 2\Sigma^2 +
\textfrac{1}{2}G_+^{-1}[1+3w +(1-w)v^2]\Omega\, .\ee

The remaining Einstein field equations together with the
Jacobi identities yield the following set of equations, which
we divide into evolution equations and constraints:

\vspace*{4mm} \noindent
\textit{Evolution equations}
\begin{subequations}\label{devoleq}
\begin{align}
\Sigma_{\alpha\beta}^\prime &= -(2-q)\Sigma_{\alpha\beta} -
2\epsilon^{\gamma\delta}{}_{\la \alpha}\,\Sigma_{\beta\ra\gamma}\,R_\delta
- {}^{3}\!{\cal R}_{\la\alpha\beta\ra} + 3\Pi_{\alpha\beta}\, ,\label{dsigevol}\\
(N^{\alpha\beta})^\prime &=
(3q\delta_\gamma{}^{(\alpha} - 2F_\gamma{}^{(\alpha})N^{\beta )\gamma}\, ,\label{dnevol}\\
A_{\alpha}^\prime &= F_\alpha{}^\beta\,A_\beta\, ,\label{daevol}\\
\Om^\prime &=  (2q-1)\,\Om - 3P +
2A_{\alpha}\,Q^{\alpha} - \Sigma_{\alpha\beta}\Pi^{\alpha\beta}\, ,\label{domevol}\\
v^\prime &= -G_-^{-1}\,(1-v^2)\,
\left[1-3w + \Sigma_{\alpha\beta}\,c^\alpha c^\beta + 2w\,(A_\beta\,c^\beta)\,v \right] v\, ,
\label{dpec}\\
c_{\alpha}^\prime &=
[\delta_\alpha{}^\beta - c_\alpha c^\beta]
[F_\beta{}^\gamma\,c_{\gamma} - v (A_\beta +
\epsilon_{\beta}{}^\gamma{}_\delta\,
N^{\delta\zeta}\,c_{\zeta}\,c_{\gamma})]\, . \label{dcpec}
\end{align}
\end{subequations}

\noindent \textit{Constraint equations}
\begin{subequations}
\begin{align}
0 &= 1 - \Sigma^2 - \Omega_{\mathrm{k}} - \Omega\, ,\label{dGauss}\\
0 &= (3\delta_\alpha{}^\gamma\,A_\beta + \epsilon_\alpha{}^{\delta\gamma}
\,N_{\delta\beta})\,\Sigma^\beta{}_\gamma - 3Q_\alpha\, ,\label{dCodazzi}\\
0 &= A_\beta\, N^\beta{}_\alpha\, ,\label{dJacobi}
\end{align}
\end{subequations}
where
$\Sigma^2=\frac{1}{6}\Sigma_{\alpha\beta}\Sigma^{\alpha\beta}$,
$A^2=A_\alpha A^\alpha$,
$F_{\alpha}{}^{\beta} = q\,\delta_{\alpha}{}^{\beta} -
\Sigma_{\alpha}{}^{\beta} -
\epsilon_{\alpha}{}^{\beta}{}_{\gamma}\,R^{\gamma}$,
and
\be\label{threecurv} {}^{3}\!{\cal R}_{\la\alpha\beta\ra} =
{\cal B}_{\la \alpha\beta \ra} -
2\epsilon^{\gamma\delta}{}_{\la
\alpha}\,N_{\beta\ra\gamma}\,A_\delta\, ;\quad
\Omega_{\mathrm{k}} = \textfrac{1}{12}{\cal B}^\alpha{}_\alpha
+ A^2\, ;\quad {\cal B}_{\alpha\beta} = 2
N_{\alpha\gamma}\,N^\gamma{}_\beta -
N^\gamma{}_\gamma\,N_{\alpha\beta}\, . \ee
The Gauss constraint~\eqref{dGauss} and the Codazzi
constraint~\eqref{dCodazzi} will figure prominently throughout.

It is of interest to note that
\be Q^\prime = -[2-q -
F_{\alpha\beta}\,c^\alpha\,c^\beta +
2(A_\alpha\,c^\alpha)\,v]\,Q\, . \label{dlqmalpha} \ee
Another quantity of interest, intimately connected with $l$, is defined
via $\tilde{n}\,\Gamma=g^{-1/2}\,l$; raising the r.h.s. with
$1+w$ and normalizing with $H$ yields the quantity\footnote{This quantity has appeared
before in the literature, e.g. in~\cite{hewetal01}, where
$\Gamma^{-(1-w)}\,G_+^{-1}$ has been denoted by $\beta$.}
\be \Psi = \Gamma^{-(1-w)}\,G_+^{-1}\,\Omega\, ,\qquad
(\ln \Psi)^\prime = 2q - (1+3w) + 2(1+w)(A_\alpha\,c^\alpha)\,v\, .\ee
%

\subsection{Bianchi type II}\label{BII}

For the Bianchi type II models we have $A_\alpha=0$, and in
addition we can choose a spatial frame $\vece_\alpha$ to be an
eigenframe of the matrix $N_{\alpha\beta}$, with
$N_{11} \neq 0$, while otherwise $N_{\alpha\beta}=0$.
This leads to that eq.~\eqref{dnevol} yields $R_2 =
\Sigma_{31}$, $R_3 = -\Sigma_{12}$, and that the Codazzi
constraint~\eqref{dCodazzi} implies
\be v_1=0=c_1\, .\label{vcond}\ee
Inserting the conditions of eq.~\eqref{vcond} into
eq.~\eqref{dcpec} gives the following relation:
\be 0 = \Sigma_{12}\,c_2 + \Sigma_{31}\,c_3 = -R_3\,c_2 +
R_2\,c_3 \qquad \Leftrightarrow \qquad \epsilon_{AB}\,R^A\,c^B
= 0\,  ,\ee
where $A,B = 2,3$ and $\epsilon_{AB}$ is the two-dimensional
permutation that has $\epsilon_{23}=1$ (hence $c_A\,c^A=1$). It
follows from $\epsilon_{AB}\,R^A\,c^B = 0$ that $R_A \propto
c_A \propto Q_A$, where the last relation holds when
$\Omega\neq 0$; hence $\Sigma_{12}$ and $\Sigma_{31}$ are
linearly dependent and can be replaced by a single variable.

We have the freedom to rotate in the 23-plane, which is
expressed in the field equations as the freedom to choose
$R_1$. To obtain a set of variables that are invariant under
such rotations we introduce the following new shear variables
\begin{subequations}
\begin{alignat}{2}
\Sigma_+ &= \tfrac{1}{2}\Sigma_A{}^A = -\tfrac{1}{2}\Sigma_{11}\, ,
&\qquad
\bar{\Sigma} &=
\tfrac{1}{\sqrt{3}}(\Sigma_{AB} - \Sigma_+\delta_{AB})
(c^{A}c^{B} - \tfrac{1}{2}\delta^{AB})\, ,\\
\tilde{\Sigma} & =\tfrac{1}{\sqrt{3}}(\Sigma_{AB} - \Sigma_+\delta_{AB})
\epsilon^B{}_C(c^{A}c^{C} - \tfrac{1}{2}\delta^{AC})\, ,
&\qquad
\check{\Sigma}^2 &= \tfrac{1}{3}\left(\Sigma_{12}^2 + \Sigma_{31}^2\right)\, .
\end{alignat}
\end{subequations}
Hewitt et al.~\cite{hewetal01} make use of the freedom to
rotate in the 23-plane to set $c_2=0$, which yields that
$c_3=1$ and $R_2=0=\Sigma_{31}$. This leads to a correspondence
between the variables $\Sigma_-,\,\Sigma_1,\,\Sigma_3$
in~\cite{hewetal01} to our variables, when setting $c_2=0$,
according to $\Sb=-\Sigma_-$, $\Sc^2=\Sigma_3^2$,
$\St^2=\Sigma_1^2$.

Because of the existence of discrete symmetries one can
simplify the analysis, e.g. the eigenvalue analysis, by
introducing the following state vector $ {\bf S} =
(\Sigma_+,\bar{\Sigma},\tilde{\Sigma}^2,\check{\Sigma}^2,\Omk,v^2)$,
where $\Omk = N_{11}^2/12$, and where $\Omega$ can be obtained
in terms of ${\bf S}$ via the Gauss constraint~\eqref{dGauss}.
This leads to the dynamical system given in
Section~\ref{sec:BII}.

\section{Hamiltonian considerations and derivation of monotone functions}\label{ham}

\subsection{Hamiltonian considerations}\label{hamconsider}

The scalar Hamiltonian is given by
\be\label{hamfkn} \Htilde = 2N\,g^{\frac{1}{2}}n^a\,n^b(G_{ab}
- T_{ab}) = 2\nt\,g\,n^a\,n^b(G_{ab} - T_{ab}) = \nt\mathscr{H}
= \nt(T + U_\mathrm{g} + U_\mathrm{f})\, ,\ee
where we have defined $\nt = N g^{-1/2}$, and where
$(T,\,U_\mathrm{g},U_\mathrm{f}) = 6\,g\,H^2(-1 +
\Sigma^2,\,\Omega_\mathrm{k},\,\Omega)$. By means
of~\eqref{nonormpfrel}, \eqref{nmudef}, and~\eqref{ldef}, this
leads to that
\be \label{Uf} U_\mathrm{f} = 2 g \rho = 2
l^{1+w}\,g^{(1-w)/2}\,\Gamma^{1-w}G_+\, , \ee
while $T$ and $U_\mathrm{g}$ depends on the model and the
metric representation.

In~\cite{uggetal95} it was shown that the tilted orthogonally
transitive Bianchi models exhibit a so-called timelike
homothetic Jacobi symmetry. It was later realized that such
symmetries are related to the existence of monotonic functions
(Uggla in ch. 10 in~\cite{waiell97}, and~\cite{heiugg10}).
Unfortunately the analysis in~\cite{uggetal95} is quite
cumbersome and we will therefore make a new derivation of the
`homothetic structure' and from this derive a monotone
function. To do so, we need to connect the spatially
homogeneous frame in~\eqref{ds2} with the orthonormal frame.
This is done in two steps: (i) diagonalization by means of the
off-diagonal special automorphism group, (ii) normalization by
means of diagonal scaling. We hence write
$e_\alpha{}^{\hat{i}}$ in the transformation~\eqref{framerel},
i.e., $\vector{e}_\alpha =
e_\alpha{}^{\hat{i}}(x^0)\,\hat{\vector{e}}_{\hat{i}}$, as
$e_\alpha{}^{\hat{i}} =
(D^{-1})_\alpha{}^{\tilde{j}}(S^{-1})_{\tilde{j}}{}^{\hat{i}}$,
or $e^\alpha{}_{\hat{i}}
=D^\alpha{}_{\tilde{j}}\,S^{\tilde{j}}{}_{\hat{i}}$, where
$S^{\tilde{j}}{}_{\hat{i}}$, since it is assumed to be a
special automorphism transformation, leaves $a_{\hat{i}}$ and
$n^{\hat{i}\hat{j}}$ unaffected. In addition we define the new
metric variables $\beta^1, \beta^2, \beta^3$ via the matrix
$D^\alpha{}_{\tilde{j}}$ so that
\begin{equation}\label{diagmatrix}
(D^{-1})_\alpha{}^{\tilde{j}} =
\begin{pmatrix}
\exp(-\beta^1) & 0 & 0 \\
0 & \exp(-\beta^2) & 0 \\
0 & 0 & \exp(-\beta^3)
\end{pmatrix}
\:,
\end{equation}
where $\beta^\alpha = \beta^\alpha(x^0)$; hence
$g^{1/2} = \exp(\beta^1 + \beta^2 + \beta^3)$.

To obtain the Hamiltonian for the tilted orthogonally
transitive Bianchi models we choose a spatially homogeneous
frame so that the line element can be written as~\eqref{ds2}
with all structure constants being zero except
$n^{\hat{1}\hat{1}}=\hat{n}_1$. In the orthogonally transitive
case we can specify the spatially homogeneous frame so that
$g_{\hat{i}\hat{j}}$ in~\eqref{ds2} have one off-diagonal
component, $g_{\hat{1}\hat{2}}$, and so that the perfect fluid
velocity has a single non-zero component, $v_{\hat{3}}$. Hence
we follow~\cite{uggetal95} and write\footnote{There is a
typographical error in eq. (2.61) in~\cite{uggetal95}; the
exponent should be $-1$ and not $-1/2$ of $(\hat{m}^{(3)})$ in
the expression for $s_3$.}
\begin{equation}
S^\alpha{}_{\tilde{j}} =
\begin{pmatrix}
1 & -\sqrt{2}\hat{n}_1\theta^3(x^0) & 0 \\
0 & 1 & 0 \\
0 & 0 & 1
\end{pmatrix}
\:.
\end{equation}
It follows from~\eqref{n} that
\be\label{nbeta}
n^{11} = \exp(\beta^1 - \beta^2 - \beta^3)\hat{n}_1\, .
\ee
Due to that $\theta_3$ is associated with the off-diagonal
special automorphism group, the associated momentum, which is
proportional to $\sigma_{12}$, is conserved~\cite{uggetal95}.
Furthermore, note that~\eqref{tcons} yields the constant of the
motion $t_{\hat{3}}=\hat{t}_3=\mathrm{const}$. This is
consistent with the Codazzi constraint, also known as the
Hamiltonian momentum constraint, which linearly relates these
two constants to each other. The constants $\hat{n}_1$, $l$,
$\hat{t}_3$ allows us to write $T + U_\mathrm{g} +
U_\mathrm{f}$ in~\eqref{hamfkn} so that
\be
\mathscr{H} = T + U_\mathrm{g} + U_\mathrm{f} =
T_\mathrm{d} + U_\mathrm{c} + U_\mathrm{g} + U_\mathrm{f}\, ,
\ee
where $U_\mathrm{c}$ is the so-called centrifugal potential, which is proportional
to the $\Sigma_{12}^2$ term in $\Sigma^2$, where, furthermore,
$U_\mathrm{c}$, $U_\mathrm{g}$, and $U_\mathrm{f}$ are all expressible in terms
of $\beta^\alpha$, and no other time dependent quantities.
Expressing $\nt T_\mathrm{d}$ in terms of the $\dot{\beta}^\alpha$ or the
associated momenta $\pi_\alpha$, see~\cite{heiugg10} and also~\cite{dametal03},
by means of~\eqref{onquant} and the transformations in this subsection, leads to
\be
\nt T_\mathrm{d} = 2\nt^{-1}\,\cg_{{\gamma} \delta}\dot{\beta}^\gamma\dot{\beta}^\delta =
\textfrac{1}{4}\nt \,\cg^{{\gamma} \delta}\pi_{\gamma}\pi_{\delta}\, ,
\ee
where $\cg_{{\gamma} \delta}$ is known as the minisuperspace metric for the
diagonal degrees of freedom, which, together with its inverse $\cg^{\alpha\beta}$
is given by
\begin{equation}\label{minisupmetric}
\cg_{{\alpha} \beta} =
\begin{pmatrix}
0 & -1 & -1 \\
-1 & 0 & -1 \\
-1 & -1 & 0
\end{pmatrix}
;\qquad\quad
\cg^{{\alpha} \beta} = \textfrac{1}{2}\begin{pmatrix}
1 & -1 & -1 \\
-1 & 1 & -1 \\
-1 & -1 & 1
\end{pmatrix}\:.
\end{equation}

The centrifugal potential
$U_\mathrm{c}= 2\,g\,H^2\,\Sigma_{12}^2 = 2\,g\,\sigma_{12}^2$
re-expressed via the Codazzi constraint in terms of $\hat{t}_3$
yields
\be
U_\mathrm{c} \propto \exp[-2(\beta^1 - \beta^2)]\, .
\ee
By means of $U_\mathrm{g} = 6\,g\,H^2\,\Omega_\mathrm{k}$
and~\eqref{nbeta} we find that
\be
U_\mathrm{g}\propto \exp(4\beta^1)\, .
\ee

We can write $\tilde{u}^a\,\tilde{u}_a=-1$ as $0=1 - \Gamma^2 +
\tilde{\mu}^{-2}\,t_\alpha\,t^\alpha$, and by defining
\be F = (1+w)^{-2}\,l^{-2w}\,g^{w}\,t_\alpha\,t^\alpha =
(1+w)^{-2}\,l^{-2w}\,g^{w}\,g^{\hat{i}\hat{j}}\,t_{\hat{i}}\,t_{\hat{j}}\, ,
\ee
we obtain
\be\label{FGamma} 0 = 1 - \Gamma^2 + F\,\Gamma^{2w}\, ,
\ee
which allows one to, in general, implicitly express $\Gamma^2$
in terms of $F$, i.e., $\Gamma^2 = \Gamma^2(F)$. In the present case
we find that
$F=(1+w)^{-2}\,l^{-2w}\,g^{w}\,g^{\hat{3}\hat{3}}\,t_{\hat{3}}\,t_{\hat{3}}$,
which yields
\be \label{F} F =
(1+w)^{-2}\,l^{-2w}\,\hat{t}_3^2\,\exp[w(\beta^1 + \beta^2) -
(1-w)\beta^3]\, . \ee
It follows that
\be U_\mathrm{f} \propto \exp[(1-w)(\beta^1 + \beta^2 +
\beta^3)]\phi . \ee
where $\phi := \Gamma^{1-w}G_+$ is a function of the particular
combination $w(\beta^1 + \beta^2) - (1-w)\beta^3$ only.
\subsection{Derivation of monotone functions}\label{dermon}

Based on the Hamiltonian for the diagonal degrees of freedom
for the orthogonally transitive type II case we in this
subsection derive a monotonic function that is of key
importance for understanding the dynamics. In ch.~10
in~\cite{waiell97} and in~\cite{heiugg10} it is shown that
monotone functions are associated with `homothetic' symmetries
of the potential, i.e., we require that there exists a vector
${\bf c} = c^\alpha\partial_{\beta^\alpha}$ such that ${\bf c}
U = {\bf c}(U_\mathrm{c} + U_\mathrm{g} + U_\mathrm{f}) =
r\,U$, where $r$ is a constant. For this to be possible we
require (i) that ${\bf c}\phi = 0$, and hence that
\be {\bf c}(w(\beta^1 + \beta^2) - (1-w)\beta^3) = w(c^1 + c^2)
- (1-w)c^3 = 0\, , \ee
and (ii) ${\bf c}U_\mathrm{c} = r\, U_\mathrm{c}$,
${\bf c}U_\mathrm{g} = r\, U_\mathrm{g}$, and
${\bf c}U_\mathrm{f} = r\, U_\mathrm{g}$, which due to the condition (i)
yields
${\bf c}\exp[(1-w)(\beta^1 + \beta^2 + \beta^3)] = r\,\exp[(1-w)(\beta^1 + \beta^2 + \beta^3)]$.
This leads to
\be
 -2(c^1-c^2) = r\, ,\qquad 4c^1 = r\, ,\qquad (1-w)(c^1 + c^2 + c^3) = r\, ,
\ee
which yields $c^2 = 3c^1$, $c^3 = 4c^1 w/(1-w)$; w.l.g
we can choose $c^1=1-w$, which gives
\be\label{ccomp} (c_1,c_2,c_3) = (1-w,3(1-w), 4w)\, ,\qquad r =
4(1-w)\, . \ee
The causal character of this vector with respect to the metric
$\cg_{{\alpha} \beta}$ is crucial~\cite{heiugg10}; we obtain
\be
\cg_{{\gamma} \delta}c^\gamma\,c^\delta = -2(1-w)(3 + 13w)\, ,
\ee
which is timelike if $-3/13<w<1$. The above properties of the potential
$U$ implies that the model satisfies the criteria given in~\cite{heiugg10}
for admitting a monotonic function given by\footnote{This is actually the
square of the monotone function in~\cite{heiugg10}, but we find this
form more convenient in the present context.}
\be
M \propto (c^\alpha\,\pi_\alpha)^2\exp\left[-r\,
\frac{\cg_{{\gamma} \delta}c^\gamma\,\beta^\delta}{\cg_{{\gamma} \delta}c^\gamma\,c^\delta} \right]\, .
\ee

We now have to express $M$ in the dynamical systems variables
${\bf S}$. We do so in two steps by first using
$\Sigma_{\alpha\alpha}=\Sigma_\alpha$, and then in the second
step we go over to the presently used dynamical systems
variables via the current gauge fixing (see the previous
discussion about the correspondence between our variables and
those used by~\cite{hewetal01}). Let us define
\be \label{Vdef} (V_\mathrm{g},\,V_\mathrm{f}) = \exp\left[-r\,
\frac{\cg_{{\gamma} \delta}c^\gamma\,\beta^\delta}
{\cg_{{\gamma} \delta}c^\gamma\,c^\delta}
\right](\,U_\mathrm{g},\,U_\mathrm{f})\, , \ee then
\be
\exp\left[-r\,
\frac{\cg_{{\gamma} \delta}c^\gamma\,\beta^\delta}{\cg_{{\gamma} \delta}c^\gamma\,c^\delta}
\right]\,\mathscr{H} =
\exp\left[-r\,
\frac{\cg_{{\gamma} \delta}c^\gamma\,\beta^\delta}{\cg_{{\gamma} \delta}c^\gamma\,c^\delta}
\right]\,T + V_\mathrm{g} + V_\mathrm{f} = 0\, ,
\ee
which yields
\be \exp\left[-r\, \frac{\cg_{{\gamma}
\delta}c^\gamma\,\beta^\delta}{\cg_{{\gamma}
\delta}c^\gamma\,c^\delta} \right] = -\left(\frac{V_\mathrm{g}
+ V_\mathrm{f}}{T}\right) \propto \frac{V_\mathrm{g} +
V_\mathrm{f}}{\pi_0^2(1-\Sigma^2)}\, , \ee
and hence
\be
M \propto \left(\frac{c^\alpha\pi_\alpha}{\pi_0}\right)^2
\frac{V_\mathrm{g} + V_\mathrm{f}}{(1-\Sigma^2)} =
\left(\frac{c^\alpha\pi_\alpha}{\pi_0}\right)^2\, V_\mathrm{f}\,\Omega^{-1}\, ,
\ee
since $V_\mathrm{g}/V_\mathrm{f} = \Omk/\Omega$ and $1-\Sigma^2 = \Omk + \Omega$.

It follows from our definitions that $\pi_\alpha =
\tfrac16\,\pi_0\,(2-\Sigma_\alpha)$, $\pi_0 = \pi_1 + \pi_2 +
\pi_3$, see also~\cite{heiugg10,heietal09}. In combination
with~\eqref{ccomp} we find that this yields
\be \frac{c^\alpha\pi_\alpha}{\pi_0} \propto 1 -
\frac18(1-w)\Sigma_1 - \frac38(1-w)\Sigma_2 - \frac12 w\Sigma_3
= 1 - \frac{\sqrt{3}}{8}(7w-3)\Sb - \frac18(1+3w)\Sp\, . \ee

Next we need to solve for $V_\mathrm{f}$ in terms of the state
space variables $(\Sigma_+, \bar{\Sigma}, \check{\Sigma}^2,
\Omk,v^2)$. This can be done with the help of the constants of
motion $\hat{n}_1,\, l,\, \hat{t}_3$ through equations
\eqref{Vdef}, \eqref{F},~\eqref{FGamma}, \eqref{nbeta}, and
\eqref{Uf}, using the form \eqref{ccomp} of the homothetic
vector {\bf c}. After some algebra one finds that
\be
V_\mathrm{f} \propto \left[(v^2)^{3-7w}\,\Gamma^{8(1-w)}\,G_+^{2(1+7w)}
\left(\frac{\Omk}{\Omega}\right)^{-(1-w)}\right]^{1/(3+13w)}\, .
\ee
Taking $M^{-(3+13w)/2}$ and replacing $\Omega$ with $\Psi$
via~\eqref{Psidef}, and $v^2$ through eq.~\eqref{codazzi}
yields the monotone function $M_\mathrm{H}$, which we give in
the main text. Similar, but much simpler, Hamiltonian methods
were used to find the monotone function $M_\mathrm{CS}$ for the
non-tilted Bianchi type II models, which is also given in the
main text.

\end{appendix}

\end{document}